\documentclass[conference,10pt]{IEEEtran}
\usepackage{cite}
\usepackage{array}
\usepackage{amssymb,amsmath}		
\usepackage{graphics}		
\usepackage{longtable}          
\usepackage{subfigure}
\usepackage{float}
\newcommand{\vectornorm}[1]{{\left|\left|#1\right|\right|}_{2}} 
 
\newcommand{\vectornormZero}[1]{{\left|\left|#1\right|\right|}_{0}} 
\newtheorem{theorem}{Theorem}[section]

\newenvironment{proof}[1][Proof]{\begin{trivlist}
\item[\hskip \labelsep {\bfseries #1}]}{$\blacksquare$ \end{trivlist}}

\begin{document}
\title{Restricted Isometry Property in \\ Quantized Network Coding of Sparse Messages}
\author{\IEEEauthorblockN{Mahdy Nabaee and Fabrice Labeau}
\IEEEauthorblockA{Electrical and Computer Engineering Department, McGill University, Montreal, QC
}
}

\maketitle

\begin{abstract}
In this paper, we study joint network coding and distributed source coding of inter-node dependent messages, with the perspective of compressed sensing.
Specifically, the theoretical guarantees for robust $\ell_1$-min recovery of an under-determined set of linear network coded sparse messages are investigated.
We discuss the guarantees for $\ell_1$-min decoding of quantized network coded messages, based on Restricted Isometry Property (RIP) of the resulting measurement matrix.
This is done by deriving the relation between tail probability of $\ell_2$-norms and satisfaction of RIP.
The obtained relation is then used to compare our designed measurement matrix, with i.i.d. Gaussian measurement matrix, in terms of RIP satisfaction.
Finally, we present our numerical evaluations, which shows that the proposed design of network coding coefficients results in a measurement matrix with an RIP behavior, similar to that of i.i.d. Gaussian matrix.

\end{abstract}

\IEEEpeerreviewmaketitle

\section{Introduction}

Efficient data gathering in sensor networks has been the topic of many research projects where different applications have been considered.
One of the concerns in data gathering is to take care of \textit{inter-node redundancy} during the transmission.
When the knowledge of inter-node dependency is known at the encoders (\textit{i.e.} sensor nodes) and the decoder node(s), distributed source coding \cite{xiong2004distributed} and optimal packet forwarding is the best transmission method, in terms of achieved information rates \cite{NCCorr_NIFwithCOrrSo}.
However, flexibility and robustness to network changes, and
no need (et the encoders) to the knowledge of inter-node dependency has drawn attention to random linear network coding \cite{NC_RLNCtoMulticast} as an alternative transmission method \cite{Ho04networkcoding}.

Recently, the concepts of \textit{compressed sensing} \cite{1614066} have been used to perform an embedded distributed source coding in linear network coding of correlated or sparse messages
\cite{rabbat,CdataGathering,sFeiz,feizi2011power}.
Joint source, channel, and network coding is studied in \cite{sFeiz,feizi2011power}, where analogue network coding \cite{katti2007embracing} is used as a linear mapping to decrease temporal and spatial redundancy of sensor data.
In \cite{naba}, we proposed Quantized Network Coding (QNC) with $\ell_1$-min decoding , where the sparse messages can be recovered from smaller number of packets compared to the conventional linear network coding \cite{NC_RLNCtoMulticast}.

To guarantee robust $\ell_1$-min recovery of messages from an under-determined set of linear measurements, the total measurement matrix has to be appropriate (or in other words satisfy some special properties).
For instance, if it satisfies Restricted Isometry Property (RIP) of appropriate order, then $\ell_1$-min recovery is feasible \cite{candes,LinProg}.
However, the literature of compressed sensing-based network coding does not include any result discussing theoretical (or even practical) requirements for robust $\ell_1$-min recovery of linear network coded messages.

In this paper, we discuss theoretical guarantees for $\ell_1$-min decoding of quantized network coded messages, based on RIP.
Specifically, we discuss the satisfaction of RIP and its implications for the measurement matrix, resulting from the design of local network coding coefficients, proposed in \cite{naba}.

The description of data gathering scenario and formulation of our proposed quantized network coding \cite{naba} is presented in section~\ref{sec:QNC}.
This is followed by a discussion on choosing appropriate local network coding coefficients, which result in zero mean Gaussian entries for the measurement matrix, in section~\ref{sec:DesignNCodes}.
In section~\ref{sec:TailProbRIP}, we derive the relation between the tail probability of $\ell_2$-norms and satisfaction of RIP, and discuss satisfaction of RIP for our designed measurement matrices.
In section~\ref{sec:Numerical}, a numerical example is presented, which compares the measurement matrix, resulting from our QNC scenario with the case of perfect Gaussian measurement matrix.
Finally, in section~\ref{sec:Conclusions}, we discuss our concluding remarks on satisfaction of RIP in our QNC scenario.

\section{Quantized Network Coding with $\ell_1$-min Decoding in Lossless Networks}\label{sec:QNC}

In this paper, we consider a lossless sensor network, represented by a directed graph, $\mathcal{G}=(\mathcal{V},\mathcal{E})$, where $\mathcal{V}=\{1,\ldots,n\}$ is the set of its nodes.
$\mathcal{E}=\{1,\ldots,|\mathcal{E}|\}$ is also the set of edges (links), where each edge, $e \in \mathcal{E}$, maintains a lossless communication from $tail(e)$ node to $head(e)$ node, at a maximum rate of $C_e$ bits per link use.
As a result, the input content of edge $e$ at time $t$, represented by $y_e(t)$ (since the links are lossless, input and output contents of each edge are the same), is from a discrete finite alphabet of size $2^{L C_e}$.
Time index, $t$, is integer and a time unit represents the time in which blocks of length $L$ are transmitted over all edges.
The sets of incoming and outgoing edges of node $v$, are defined respectively:
\begin{eqnarray}
\textit{In}(v)&=&\{e: head(e)=v, e \in \mathcal{E}\}, \nonumber \\
\textit{Out}(v)&=&\{e: tail(e)=v, e \in \mathcal{E}\}. \nonumber
\end{eqnarray}
We assume that each node $v$ has a random information source, $X_v$, which generates (random) message, called $x_v$, where $x_v \in \mathbb{R}$.
Furthermore, we consider the case where the messages, $\underline{x}=[x_v:v \in \mathcal{V}] \in \mathbb{R}^n$, are such that there is a linear transform matrix, $\phi_{n \times n}$, for which $\underline{x}=\phi \cdot \underline{s}$, and $\underline{s}$ is $k$-sparse (has at most $k$ non-zero elements).
In (single session) data gathering, all the messages, $x_v$'s, are to be transmitted to a single gateway (or decoder) node, represented by $v_0$, where $v_0 \in \mathcal{V}$.

QNC at each node, $v \in \mathcal{V}$, was defined in \cite{naba}, as follows:
\begin{equation}\label{Eq:QNC1}
y_e(t)= \textbf{Q}_e \Big [\sum_{e' \in \textit{In}(v)} \beta_{e,e'}(t)\cdot y_e(t-1)+\alpha_{e,v}(t)\cdot x_v \Big ], 
\end{equation}
where $\textbf{Q}_e[ \centerdot]$ is the quantizer (designed based on the value of $C_e$ and the distribution of incoming contents and messages), associated with the outgoing edge $e \in \textit{Out}(v)$, and $\beta_{e,e'}(t)$ and $\alpha_{e,v}(t)$ are the corresponding network coding coefficients, picked from real numbers.
Initial rest condition is also assumed to be satisfied in our QNC scenario: $y_e(1)=0,~\forall~e \in \mathcal{E}.$
We represent the quantization error of $\textbf{Q}_e[ \centerdot]$ by $n_e(t)$, which implies:
\begin{equation}\label{Eq:QNC2}
y_e(t)= \sum_{e' \in \textit{In}(v)} \beta_{e,e'}(t)\cdot y_e(t-1)+\alpha_{e,v}(t)\cdot x_v + n_e(t). 
\end{equation}
Equivalently, we have \cite{naba}:
\begin{equation}\label{Eq:matrixForm}
\underline{y}(t)=F(t) \cdot \underline{y}(t-1)+A(t) \cdot \underline{x}+\underline{n}(t), 
\end{equation}
where $\underline{y}(t)=[y_e(t):e \in \mathcal{E}]$, $\underline{n}(t)=[n_e(t):e \in \mathcal{E}]$, and,
\begin{equation}
 F(t)_{|\mathcal{E}|\times |\mathcal{E}|}: \{F(t)\}_{e,e'}=\left\{
\begin{array}{l l}
  \beta_{e,e'}(t)  & ,~\scriptsize tail(e)=head(e') \\
  0  &  ,~\mbox{otherwise} \\ \end{array} \right. \nonumber
\end{equation}
\begin{equation}\label{Eq:defineAt}
 A(t)_{|\mathcal{E}|\times |\mathcal{V}|}: \{A(t)\}_{e,v}=\left\{
\begin{array}{l l}
  \alpha_{e,v}(t)  & ,~tail(e)=v \\
  0  &  ,~\mbox{otherwise} \\ \end{array} \right. . \nonumber
\end{equation}
By using linearity in the QNC scenario, the \textit{marginal measurements} at time $t$, represented by $\{z(t)\}_i$'s, where $\underline{z}(t)=[y_e(t):e \in \textit{In}(v_0)]$, are calculated as:
\begin{equation}\label{Eq:measForm1}
\underline{z}(t)=B(t) \cdot \underline{y}(t)= \Psi(t) \cdot \underline{x}+\underline{n}_{eff}(t).
\end{equation}
In Eq.~\ref{Eq:measForm1}, $\Psi(t)$ and $\underline{n}_{eff}(t)$ are defined as:
\begin{eqnarray}
\Psi(t)&=&B(t) \cdot \sum_{t'=2}^{t} \prod_{t''=t}^{t'+1} F(t'') \cdot A(t'), \label{Eq:DefPsi} \\
\underline{n}_{eff}(t)&=&B(t) \cdot \sum_{t'=2}^{t} \prod_{t''=t}^{t'+1} F(t'') \cdot \underline{n}(t'), \label{Eq:DefNeff}
\end{eqnarray}
and $B(t)$ is defined such that:
\begin{equation}
\{B(t)\}_{i,e}=\left\{
\begin{array}{l l}
  b_{i,e}(t)  & ,~i~\mbox{corresponds to}~e,~e \in \textit{In}(v_0) \\
  0  &  ,~\mbox{otherwise} \\ \end{array} \right. . \nonumber
\end{equation}

We store marginal measurements, at the decoder, and build up \textit{total measurements vector}, called $\underline{z}_{tot}(t)$, as follows:
\begin{equation}\label{Eq:totMeasEq}
\underline{z}_{tot}(t)=\left[ {\begin{array}{*{20}c}
	\underline{z}(2) \\	
   \vdots   \\
   \underline{z}(t)   \\
 \end{array} } \right]_{m \times 1},
\end{equation}
where 
\begin{equation}
m=(t-1) |\textit{In}(v_0)| ,
\end{equation}
and for which we have \cite{naba}:
\begin{equation}\label{Eq:measEq}
\underline{z}_{tot}(t)= \Psi_{tot}(t) \cdot \underline{x} + \underline{n}_{eff,tot}(t),
\end{equation}
where the \textit{total measurement matrix}, $\Psi_{tot}(t)$, and \textit{total effective noise vector}, $\underline{n}_{eff,tot}(t)$, are calculated as follows:
\begin{equation}
\Psi_{tot}(t)=\left[ {\begin{array}{*{20}c}
	\Psi(2) \\	
   \vdots   \\
   \Psi(t)   \\
 \end{array} } \right],~~
\underline{n}_{eff,tot}(t)=\left[ {\begin{array}{*{20}c}
	\underline{n}_{eff}(2) \\	
   \vdots   \\
	\underline{n}_{eff}(t)   \\
 \end{array} } \right].
\end{equation}

Since (\ref{Eq:measEq}) is in the form of a noisy linear measurement equation, compressed sensing decoding (\textit{i.e.} $\ell_1$-min recovery of Eq.~14 in \cite{naba}) can be applied, even if $m$ is smaller than $n$.
However, robust $\ell_1$-min decoding requires the total measurement matrix, $\Psi_{tot}(t)$, to satisfy some conditions \cite{LinProg,candes}.
Specifically, to ensure that the upper bound of Eq.~15 in \cite{naba} holds, we have to investigate the satisfaction of RIP for $\Psi_{tot}(t)$, in our QNC scenario.
In \cite{naba}, we proposed an appropriate design for network coding coefficients, which resulted in improved delay-quality performance for our QNC, compared to conventional packet forwarding.
In this paper, we analyze the satisfaction of RIP for $\Psi_{tot}(t)$, resulting from the proposed design of local network coding coefficients in \cite{naba} (also described in section~\ref{sec:DesignNCodes}).

\section{Design of Network Coding Coefficients}\label{sec:DesignNCodes}

Matrices with good norm conservation property are shown to be good choices for measurement in compressed sensing \cite{candes}.
RIP characterizes the norm conservation such that an $m \times n$ matrix, $\Theta$, is said to satisfy \textit{RIP of order $k$ with constant $\delta_k$}, if:
\begin{equation}\label{Eq:RIPexact}
1-\delta_k \leq \frac{ \vectornorm{\Theta \cdot \underline{s} }^2}{\vectornorm{\underline{s}}^2} \leq 1+ \delta_k,~\forall~\underline{s} \in \mathbb{R}^n,~\vectornormZero{\underline{s}} \leq k.
\end{equation}
Random matrices with Independently and Identically Distributed (i.i.d.) zero mean Gaussian entries are proved to satisfy RIP with an overwhelming probability \cite{simpleProof}.
Explicitly, for an $m \times n$ matrix with i.i.d. zero mean Gaussian entries of variance $\frac{1}{m}$, called $G_{m \times n}$, the probability of satisfying RIP of order $k$ and constant $\delta_k$ is exceeding
\begin{equation}\label{Eq:GoverWprob}
1-e^{-\kappa_2 m},
\end{equation}
(also referred as \textit{overwhelming}) where 
\begin{equation}\label{Eq:GnumMeas}
m > \kappa_1 k \log(\frac{n}{k}),
\end{equation}
and $\kappa_1, \kappa_2$ only depend on $\delta_k$ (theorem~5.2 in \cite{simpleProof}).
In the following, we mention a design (originally proposed in \cite{naba}) for the local network coding coefficients, $\beta_{e,e'}(t)$ and $\alpha_{e,v}(t)$, which results in zero mean Gaussian entries for $\Psi_{tot}(t)$.
Then in section~\ref{sec:TailProbRIP}, we derive an upper probability bound for satisfying RIP in our QNC scenario with the proposed coefficients.

\begin{theorem}\label{th:Gaussian1}
Consider a quantized network coding scenario, in which the network coding coefficients, $\alpha_{e,v}(t)$ and $\beta_{e,e'}(t)$, are such that:
\begin{itemize}
\item $\alpha_{e,v}(t)=0,~\forall t>2,$
\item $\alpha_{e,v}(2)$'s are independent zero mean Gaussian random variables,
\item $\beta_{e,e'}(t)$'s are deterministic.
\end{itemize}
For such a scenario, the entries of the resulting $\Psi_{tot}(t)$ are zero mean Gaussian random variables, and the entries of different columns of $\Psi_{tot}(t)$, \textit{i.e.} $\{\Psi_{tot}(t)\}_{iv}$ and $\{\Psi_{tot}(t)\}_{i'v'}$, where $v,v' \in \mathcal{V},~v \neq v',$ are independent.
\end{theorem}

\begin{proof}
By choosing $\alpha_{e,v}(t)=0,~\forall~t>2,$ we have:
\begin{equation}
\Psi(t)=B(t) \cdot F(t) \cdots F(3) \cdot A(2), \label{Eq:Psit}
\end{equation}
which implies that each entry of $\Psi(t)$'s and also $\Psi_{tot}(t)$ is a linear combination of entries of $A(2)$.
Moreover, since entries of $A(2)$ are zero mean Gaussian random variables, then the entries of $\Psi(t)$'s and also $\Psi_{tot}(t)$ are zero mean Gaussian random variables.
Since entries in different columns of $\Psi_{tot}(t)$, are linear combinations of two independent sets of random variables, \textit{i.e.} entries of $A(2)$, then they are also independent.
However, such conclusion can not be made for entries of the same column of $\Psi_{tot}(t)$.
\end{proof}

%We also say that $\beta_{e,e'}(t)$'s are \textit{orthogonal} sets, if for all $v \in \mathcal{V}$, we have:
%\begin{equation}
%\sum_{e'' \in \textit{In}(v)} \beta_{e,e''}(t) \cdot \beta_{e',e''}(t) = 0, \forall e,e' \in \textit{Out}(v),~e \neq e'.
%\end{equation}

\section{RIP Analysis and Tail Probability of $\ell_2$-norms}\label{sec:TailProbRIP}

Satisfaction of RIP for random matrices is usually characterized by its probability (or its lower probability bounds) \cite{simpleProof}.
Moreover, to approach the probabilistic satisfaction of RIP, we first need to derive an expression for the tail probability of $\ell_2$ norms \cite{simpleProof,CSbook}. 
Specifically, a well behaved $\Psi_{tot}(t)$ (\textit{i.e.} $\Psi_{tot}(t)$ with high RIP probability) should be such that
\begin{equation}\label{Eq:defineTailProb1}
\textbf{P}(\Big |\vectornorm{{\Psi_{tot}(t) \cdot \underline{x}}}^2-1 \Big| \geq \epsilon),
\end{equation}
is very small, for all $\underline{x}$, with $\vectornorm{\underline{x}}=1$. 
In the following, we calculate this tail probability for our QNC scenario with the proposed network coding coefficients, and then present a theorem which explicitly describes the relation between the satisfaction of RIP and the tail probability of Eq.~\ref{Eq:defineTailProb1}.
In the rest of this section, we assume that the conditions of Theorem~\ref{th:Gaussian1} hold.

Consider 
\begin{equation}
\underline{z}'=\Psi_{tot}(t) \cdot \underline{x},
\end{equation}
where $\underline{x} \in \mathbb{R}^n$, and $\vectornorm{\underline{x}}=1$.
Since the conditions of Theorem~\ref{th:Gaussian1} are satisfied, Eq.~\ref{Eq:Psit} holds, and therefore:
\begin{equation}
\Psi_{tot}(t) = \left[ {\begin{array}{*{20}c}
	\Psi(2) \\	
   \vdots   \\
   \Psi(t)   \\
 \end{array} } \right]
 = \Omega(t) \cdot A(2),
 \end{equation}
where
\begin{equation}
 \Omega(t)=\left[ {\begin{array}{*{20}c}
	B(2) \\	
	B(3) F(3) \\
   \vdots   \\
   B(t) F(t) \cdots F(3)  \\
 \end{array} } \right].
\end{equation}
This implies:
\begin{equation}
\underline{z}'=\Omega(t) ~ A(2) \cdot \underline{x},
\end{equation}
or equivalently:
\begin{eqnarray}
z'_i&=&\sum_{v=1}^n \{\Psi_{tot}(t)\}_{iv}~ x_v \nonumber \\
&=& \sum_{v=1}^n \sum_{e=1}^{|\mathcal{E}|} \{\Omega(t)\}_{ie} ~\{A(2)\}_{ev}~ x_v. 
\end{eqnarray}
By expanding $z'^2_i$, and using the fact that $\{A(2)\}_{ev}$ is non-zero only when $tail(e)=v$, we have:
\begin{eqnarray}
\vectornorm{\underline{z}'}^2 &=& \sum_{i=1}^{m} z'^2_i \label{Eq:equation1}  \\
& = &\sum_{e=1}^{|\mathcal{E}|} \sum_{e'=1}^{|\mathcal{E}|} \gamma_{e,e'}(\underline{x})  \{A(2)\}_{e,tail(e)}  \{A(2)\}_{e',tail(e')}, \nonumber
\end{eqnarray}
where:
\begin{equation}
\gamma_{e,e'}(\underline{x}) = \sum_{i=1}^{m} \{\Omega(t)\}_{ie}~\{\Omega(t)\}_{ie'} \cdot x_{tail(e)}~x_{tail(e')}
\end{equation}
Using eigen-decomposition, (\ref{Eq:equation1}) simplifies to:
\begin{eqnarray}
\vectornorm{\underline{z}'}^2= \sum_{e=1}^{|\mathcal{E}|} \lambda_{e}(\underline{x}) \cdot \chi^2_{e},
\end{eqnarray}
where $\lambda_e(\underline{x})$'s are eigen-values of the symmetric matrix
\begin{equation}
\Gamma(\underline{x})=[\gamma_{e,e'}(\underline{x})]_{|\mathcal{E}| \times |\mathcal{E}|},
\end{equation}
and $\chi^2_e$'s are independent Chi-Square random variables of first order.

Moreover, for the characteristic function of $\vectornorm{\underline{z}'}^2$, we have:
\begin{eqnarray}
\textbf{E}[e^{j\omega \vectornorm{\underline{z}'}^2 }] &=& 
\textbf{E}[e^{j\omega \sum_{e=1}^{|\mathcal{E}|} \lambda_{e}(\underline{x}) \chi^2_{e} }] \\
&=& \prod_{e=1}^{|\mathcal{E}|} \textbf{E}[e^{j\omega \lambda_{e}(\underline{x}) \chi^2_{e} }] \label{Eq:prove111} \\
&=& \prod_{e=1}^{|\mathcal{E}|} \frac{1}{\sqrt{1-j2\omega \lambda_e(\underline{x}) }},
\end{eqnarray}
where (\ref{Eq:prove111}) is derived from independence of $\chi^2_e$'s.
By using the inverse formula of characteristic function, Eqs.~\ref{Eq:provepart1}-\ref{Eq:tailProb} can be obtained, where $\textbf{p}_{\vectornorm{\underline{z}'}^2} (\centerdot)$ is the probability density function of $\vectornorm{\underline{z}'}^2$, and (\ref{Eq:prove1}) is resulted from the integral property of the Fourier transform.
\begin{figure*}[t]
\begin{eqnarray}
\textbf{P}\Big(\Big | \vectornorm{\underline{z}'}^2 -1 \Big |> \epsilon \Big) 
&=& 1+\int_{-\infty}^{1-\epsilon} \textbf{p}_{\vectornorm{\underline{z}'}^2} (\nu) d\nu 
- \int_{-\infty}^{1+\epsilon} \textbf{p}_{\vectornorm{\underline{z}'}^2} (\nu) d\nu \label{Eq:provepart1} \\
&=& 1+\frac{1}{2\pi} \int_{-\infty}^{+\infty} \frac{\textbf{E}[e^{j\omega \vectornorm{\underline{z}'}^2 }]}{-j\omega}e^{-j\omega(1-\epsilon)} d\omega 
- \frac{1}{2\pi} \int_{-\infty}^{+\infty} \frac{\textbf{E}[e^{j\omega \vectornorm{\underline{z}'}^2 }]}{-j\omega}e^{-j\omega(1+\epsilon)} d\omega \label{Eq:prove1} \\
%&=&1-\frac{1}{\pi} \int_{-\infty}^{+\infty} \frac{\textbf{E}[e^{j\omega \vectornorm{\underline{z}'}^2 }]}{\omega}e^{-j\omega} \sin(\epsilon \omega) d\omega  \\
&=&1-\frac{1}{\pi} \int_{-\infty}^{+\infty} \frac{ e^{-j\omega} \sin(\epsilon \omega)}{\omega \prod_{e=1}^{|\mathcal{E}|} \sqrt{1-j2\omega \lambda_e(\underline{x})} } d\omega, \label{Eq:tailProb}
\end{eqnarray}
\end{figure*}
The right hand side of (\ref{Eq:tailProb}) is the expression for the tail probability of $\ell_2$-norms, for a specific $\underline{x}$, resulting from our proposed network coding coefficients.

In the following, we present Theorem~\ref{theorem:RIP1}, which clarifies the relation between the tail probability of (\ref{Eq:defineTailProb1}) and the probability of satisfying RIP \textit{for a general case}.

\begin{theorem}\label{theorem:RIP1}
Consider $\Phi$ for which we have:
\begin{eqnarray}\label{Eq:RIP1condition}
\textbf{p}_{tail}({\Phi},\epsilon)&=&\max_{\underline{x}} \textbf{P}\Big(\Big |\vectornorm{{\Phi  \cdot \underline{x}}}^2-1 \Big| \geq \epsilon \Big), \nonumber \\
&& s.t.~\vectornorm{\underline{x}}=1
\end{eqnarray}
In such case, for every orthonormal $\phi$, $\Theta=\Phi \cdot \phi$ satisfies RIP of order $k$ and constant $\delta_{k}$, with a probability exceeding,
\begin{equation}\label{Eq:lowerRIPbound1}
\textbf{p}_{RIP}\Big(\Phi,k,\delta_k \Big)=1- \left(
\begin{array}{c}
n\\
k
\end{array}
\right) (\frac{42}{\delta_k})^k ~ \textbf{p}_{tail}({\Phi},\epsilon=\frac{\delta_k}{\sqrt{2}}).
\end{equation}
\end{theorem}
\begin{proof}\footnote{Most of the proof is similar to the proof of Theorem~7.3 in \cite{CSbook}.}
To prove that RIP holds, we should show that inequality of (\ref{Eq:RIPexact}) is satisfied, for all $k$-sparse vectors, $\underline{s}$.
We only need to show it is satisfied, for vectors, with $\vectornorm{\underline{s}}=1$, since  $\vectornorm{\Phi \phi \cdot \underline{s}}$ is proportional with $\vectornorm{\underline{s}}$. 
Now, fix a set $T \subset \{1,2,\ldots,n\}$, with $|T|=k$, and let $\Gamma_T$ be the subspace of $k$-dimensional vectors, $\underline{s}_T$, spanned by columns of $\Phi$, with indexes in $T$.
According to lemma 7.5 in \cite{CSbook}, we can choose a finite set of vectors, $\underline{w}_T \in \mathcal{W}_T$, where $\mathcal{W}_T \subset \Gamma_T$ and $\vectornorm{\underline{w}} \leq 1$, such that for all $\underline{s}_T \in \Gamma_T$, with $\vectornorm{\underline{s}_T} \leq 1$, we have:
\begin{equation}
\vectornorm{\underline{s}_T-\underline{w}_T} \leq \frac{\delta_k}{14},
\end{equation}
conditioned on:
\begin{equation}
|\Gamma_T| \leq (\frac{42}{\delta_k})^k.
\end{equation}
There are $\left(
\begin{array}{c}
n\\
k
\end{array}
\right)$ different $T$'s, for which we repeat the above procedure and obtain: $$\mathcal{W}= \bigcup_T \mathcal{W}_T.$$
By using the union bound and the fact that for every $\underline{x}=\phi \cdot \underline{w}$, where $\underline{w} \in \mathcal{W}$, Eq.~\ref{Eq:RIP1condition} implies:
\begin{equation}
\textbf{P}\Big(\Big |\vectornorm{{\Phi  \cdot \underline{x}}}^2-1 \Big| \geq \epsilon \Big) \leq 
\textbf{p}_{tail}({\Phi},\epsilon).
\end{equation}
Therefore, for every $\underline{w} \in \mathcal{W}$, the inequality 
\begin{equation}\label{Eq:inequality110}
(1-\frac{\delta_k}{\sqrt{2}}) \vectornorm{\underline{w}}^2 \leq \vectornorm{\Phi \phi \cdot \underline{w}}^2 \leq (1+\frac{\delta_k}{\sqrt{2}}) \vectornorm{\underline{w}}^2,
\end{equation}
holds with a probability exceeding
\begin{equation}
1-\left(
\begin{array}{c}
n\\
k
\end{array}
\right) (\frac{42}{\delta_k})^k ~\textbf{p}_{tail}({\Phi},\frac{\delta_k}{\sqrt{2}}).
\end{equation}
The rest of the proof uses the same reasoning procedure, as in the proof of Theorem~7.3 in \cite{CSbook}.
\end{proof}

It can be concluded from Theorem~\ref{theorem:RIP1} that in order to have a good RIP satisfaction (\textit{i.e.} high upper probability bound for satisfaction of RIP), a small worst case tail probability, $\textbf{p}_{tail}(\centerdot,\frac{\delta_k}{\sqrt{2}})$, is required.
In section~\ref{sec:Numerical}, we compare $\textbf{p}_{tail}(\centerdot,\frac{\delta_k}{\sqrt{2}})$'s, corresponding to our designed $\Psi_{tot}(t)$ and i.i.d. Gaussian matrix, to numerically evaluate their RIP behaviour.

\begin{figure*}[!t]
\centering
\subfigure[$1100$ edges]{
\resizebox{!}{.33\textheight}{
\includegraphics{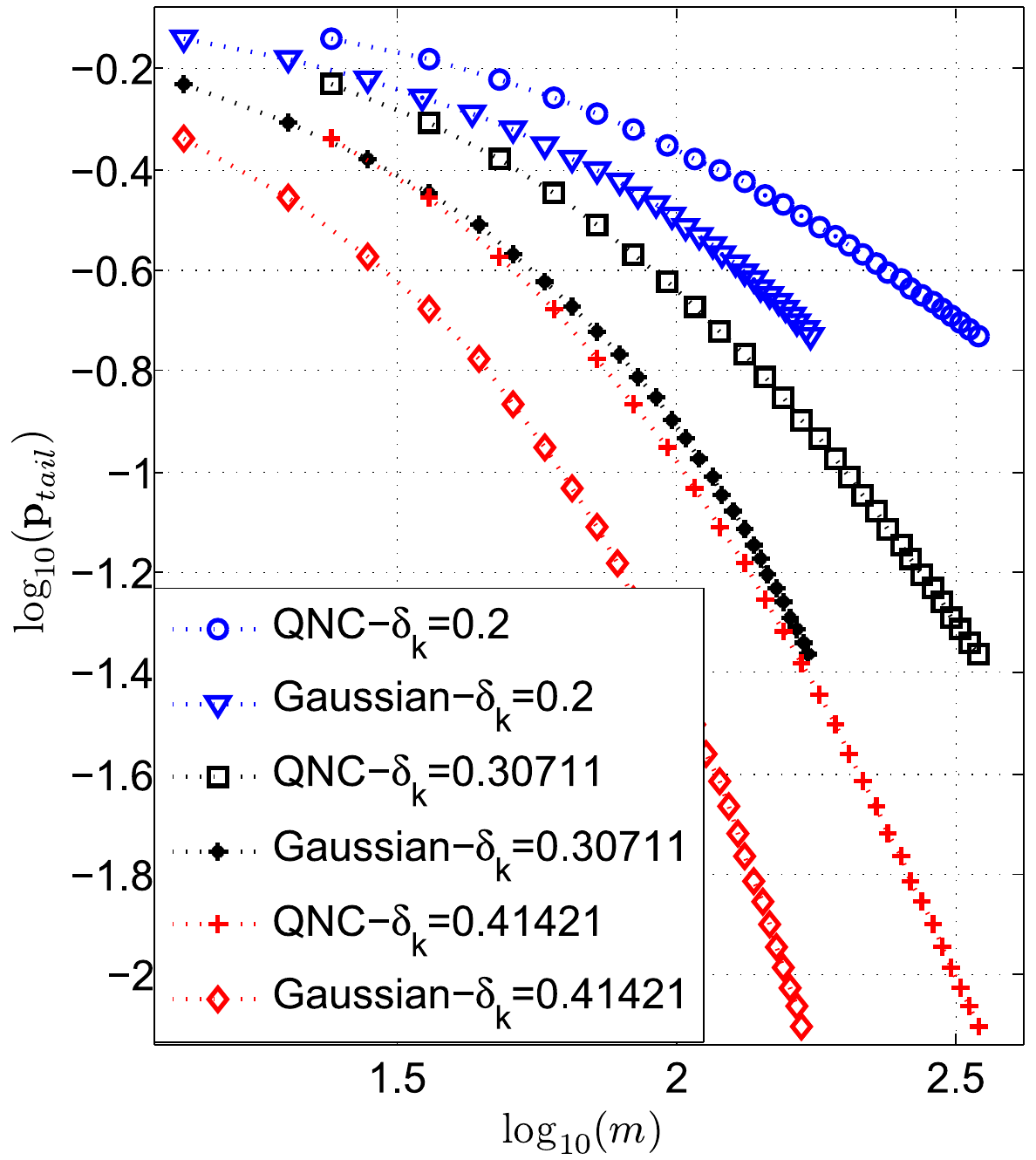}}
\label{fig:subfig1100}
} %\hspace{-.2cm}
\subfigure[$1400$ edges]{
\resizebox{!}{.33\textheight}{
\includegraphics{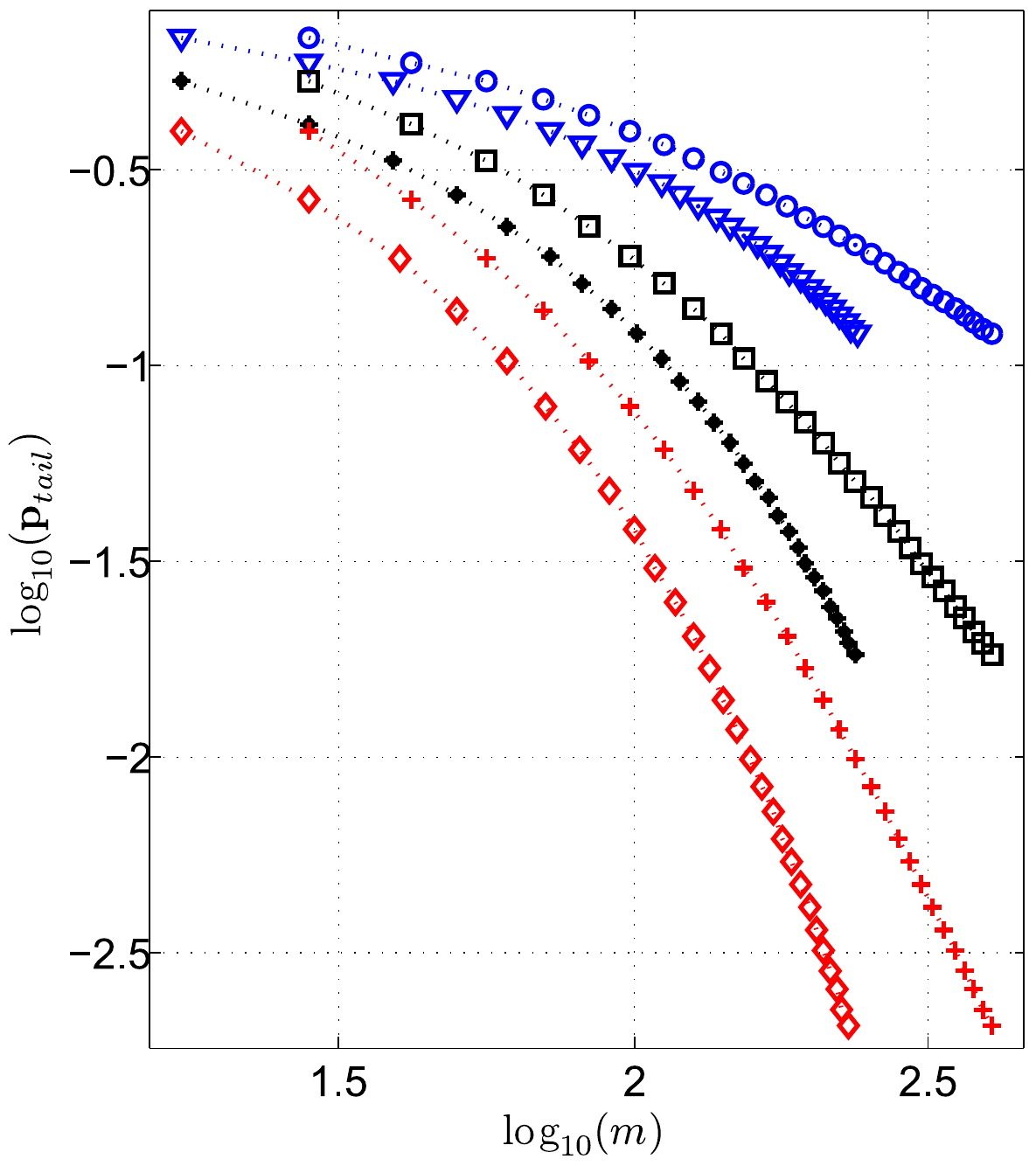}}
\label{fig:subfig1400}
} %\hspace{-.2cm}
\subfigure[$1800$ edges]{
\resizebox{!}{.33\textheight}{
\includegraphics{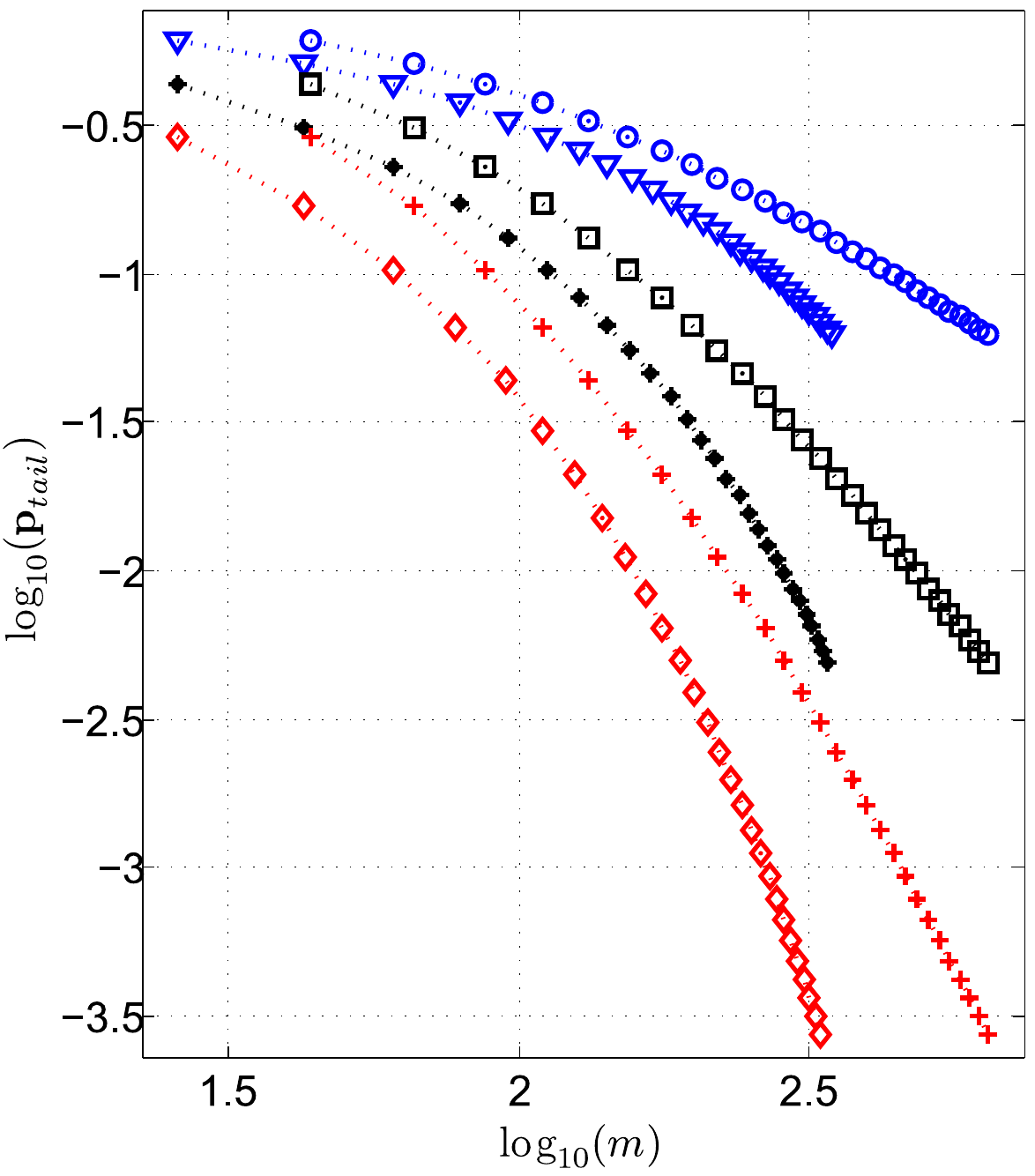}}
\label{fig:subfig1800}
} \\
\caption{Logarithmic tail probability versus logarithmic ratio of minimum required number of measurements in
our QNC scenario and i.i.d. Gaussian measurement matrices, for $n=100$, different RIP constants, and different number of edges
\label{fig:subfigureExample}.
}
\end{figure*}

%\begin{figure*}[!t]
%\centering
%\subfigure[$1100$ edges]{
%\resizebox{.32\textwidth}{!}{
%\includegraphics{tailProb1100edges.pdf}}
%\label{fig:subfig1100}
%}\hspace{-.2cm}
%\subfigure[$1400$ edges]{
%\resizebox{.32\textwidth}{!}{
%\includegraphics{tailProb1400edges.pdf}}
%\label{fig:subfig1400}
%} \hspace{-.2cm}
%\subfigure[$1800$ edges]{
%\resizebox{.32\textwidth}{!}{
%\includegraphics{tailProb1800edges.pdf}}
%\label{fig:subfig1800}
%} \\
%\caption{Logarithmic tail probability versus logarithmic ratio of minimum required number of measurements in
%our QNC scenario and i.i.d. Gaussian measurement matrices, for $n=100$, different RIP constants, and different number of edges
%\label{fig:subfigureExample}.
%}
%\end{figure*}

By using the derived tail probability of (\ref{Eq:tailProb}), and applying Theorem~\ref{theorem:RIP1}, the following theorem can be obtained, which suggests an upper probability bound on the satisfaction of RIP, in our QNC scenario.
\begin{theorem}\label{th:RIPqnc}
For a quantized network coding scenario, in which the network coding coefficients hold the conditions of Theorem~\ref{th:Gaussian1}, for every orthonormal $\phi$, the resulting $\Theta=\Psi_{tot}(t) \cdot \phi$ satisfies RIP of order $k$, and constant $\delta_k$, with a probability exceeding $$\textbf{p}_{RIP}\Big(\Psi_{tot}(t),k,\delta_k \Big),$$ defined in Eq.~\ref{Eq:definePrip}.
\begin{figure*}[]
\begin{equation}\label{Eq:definePrip}
\textbf{p}_{RIP}\Big(\Psi_{tot}(t),k,\delta_k \Big)=
1-\left(
\begin{array}{c}
n\\
k
\end{array}
\right) (\frac{42}{\delta_k})^k ~ \Big ( 1-\frac{1}{\pi} \min_{~\underline{x},~\vectornorm{\underline{x}}=1~}  \int_{-\infty}^{+\infty} \frac{ e^{-j\omega} \sin(\frac{\delta_k}{\sqrt{2}} \omega)}{\omega \prod_{e=1}^{|\mathcal{E}|} \sqrt{1-j2\omega \lambda_e(\underline{x})} } d\omega \Big)
\end{equation}
\end{figure*}
%\onecolumn
%\begin{equation}\label{Eq:definePrip}
%\textbf{p}_{RIP}\Big(\Psi_{tot}(t),k,\delta_k \Big)=
%1-\left(
%\begin{array}{c}
%n\\
%k
%\end{array}
%\right) (\frac{42}{\delta_k})^k ~ \Big ( 1-\frac{1}{\pi} \min_{~\underline{x},~\vectornorm{\underline{x}}=1~}  \int_{-\infty}^{+\infty} \frac{ e^{-j\omega} \sin(\frac{\delta_k}{\sqrt{2}} \omega)}{\omega \prod_{e=1}^{|\mathcal{E}|} \sqrt{1-j2\omega \lambda_e(\underline{x})} } d\omega \Big)
%\end{equation}
%\twocolumn
\end{theorem}

It is however difficult to derive the number of required measurements, $m$, from the expression of Eq.~\ref{Eq:definePrip}; we use numerical evaluations, in section~\ref{sec:Numerical}, to explore the properties of our QNC design.

\section{Numerical Evaluations and Discussion}
\label{sec:Numerical}
In order to evaluate the RIP satisfaction of $\Psi_{tot}(t)$, resulting from the proposed network coding coefficients, we use the worst case tail probability, $\textbf{p}_{tail}(\centerdot,\frac{\delta_k}{\sqrt{2}})$.
This is because of the deterministic (linear) relation between $\textbf{p}_{tail}(\centerdot,\frac{\delta_k}{\sqrt{2}})$ and the proposed upper probability bound in Theorem~\ref{theorem:RIP1}.
Moreover, we calculate the worst case tail probability, corresponding to an i.i.d. Gaussian matrix, called $G_{m \times n}$, and compare it with that of our $\Psi_{tot}(t)$.
For an $m \times n$ i.i.d. Gaussian matrix, $G_{m \times n}$, the worst case tail probability, $\textbf{p}_{tail}(G_{m \times n},\frac{\delta_k}{\sqrt{2}})$, can be calculated as:
\footnote{This can be obtained similar to the reasoning procedure for Eq.~\ref{Eq:tailProb}.}
\begin{equation}\label{Eq:GaussianTail}
\textbf{p}_{tail}(G_{m \times n},\frac{\delta_k}{\sqrt{2}})=
1-\frac{1}{\pi} \int_{-\infty}^{+\infty} \frac{ e^{-j\omega} \sin( \omega\frac{ \delta_k }{\sqrt{2}})}{\omega~   (1-2j\frac{\omega }{m}  )^{m/2} } d\omega. 
\end{equation}

%To present our numerical evaluations, for each value of tail probability, represented by $\textbf{p}_{tail}$, the minimum number of required measurements in $\Psi_{tot}(t)$, resulting from our QNC scenario (with the designed network coding coefficients, as in theorem~\ref{th:Gaussian1}) and $G_{m \times n}$, denoted by $M_{QNC}(\textbf{p}_{tail})$ and $M_{G}(\textbf{p}_{tail})$, are calculated:
%\begin{eqnarray}
%M_{QNC}(\textbf{p}_{tail},\delta_k)&=& \arg\min_{m} \textbf{p}_{tail}(\Psi_{tot}(t),\frac{\delta_k}{\sqrt{2}}), \\
%&& s.t.~~\textbf{p}_{tail}(\Psi_{tot}(t),\frac{\delta_k}{\sqrt{2}}) \leq \textbf{p}_{tail}, \nonumber \\
%&& m=(t-1) |\textit{In}(v_0)|  \nonumber
%\end{eqnarray}
%\begin{eqnarray}
%M_{G}(\textbf{p}_{tail},\delta_k)&=& \arg\min_{m} \textbf{p}_{tail}(G_{m \times n},\frac{\delta_k}{\sqrt{2}}), \\
%&& s.t.~~\textbf{p}_{tail}(G_{m \times n},\frac{\delta_k}{\sqrt{2}}) \leq \textbf{p}_{tail} \nonumber 
%\end{eqnarray}
%In Fig.~\ref{fig:subfigureExample}, $\textbf{p}_{tail}$ is calculated and drawn versus the ratio of minimum number of required measurements, $\frac{M_{QNC}(\textbf{p}_{tail},\delta_k)}{M_{G}(\textbf{p}_{tail},\delta_k)}$, for different values of RIP constant, $\delta_k$.
%This is done by generating random deployments of networks and calculating (\ref{Eq:tailProb}) and (\ref{Eq:GaussianTail}) in each generated deployment.
%The resulting tail probabilities, and corresponding number of measurements are then averaged over different realizations of network deployments, to obtain a set of smoothed curves in Fig.~\ref{fig:subfigureExample}.

To present our numerical evaluations, for each value of tail probability, represented by $\textbf{p}_{tail}$, the minimum number of required measurements in $\Psi_{tot}(t)$, resulting from our QNC scenario (with the designed network coding coefficients, as in Theorem~\ref{th:Gaussian1}) and $G_{m \times n}$, are calculated.
This is done by generating random deployments of networks and calculating the worst case tail probability of (\ref{Eq:tailProb}) and (\ref{Eq:GaussianTail}) in each generated deployment.
The resulting tail probabilities, and corresponding number of measurements are then averaged over different realizations of network deployments.
In Fig.~\ref{fig:subfigureExample}, $\textbf{p}_{tail}$ is drawn versus $m$ in logarithmic scale, for $\Psi_{tot}(t)$ (QNC) and $G_{m \times n}$ (Gaussian), and different values of RIP constant, $\delta_k$ ($\delta_k=0.41421 \simeq \sqrt(2)-1$ is the largest RIP constant for which Theorem~4.1 of \cite{naba} can be applied).

The statistical characteristics of the resulting $\Psi_{tot}(t)$ and its worst case tail probability vary by changing the network deployment parameters, like the distribution of edges in the network.
In Figs.~\ref{fig:subfig1100} to \ref{fig:subfig1800}, the curves correspond to different deployments with $n=100$ nodes, and $|\mathcal{E}|=1100,1400,1800$ uniformly distributed edges, respectively.
To generate the network coding coefficients, $\alpha_{e,v}(t)$'s and $\beta_{e,e'}(t)$'s, we make sure that the conditions of Theorem~\ref{th:Gaussian1} are satisfied.
Moreover, for $\beta_{e,e'}(t)$'s, it was experimentally understood that the resulting $\Psi_{tot}(t)$ has a better behavior in terms of RIP satisfaction (and also $\ell_1$-min recovery) if in any two outgoing edges, $\beta_{e,e'}$'s are orthogonal.
%
%for all $v \in \mathcal{V}$, the following \textit{orthogonality} condition is satisfied:
%\begin{equation}
%\sum_{e'' \in \textit{In}(v)} \beta_{e,e''}(t) \cdot \beta_{e',e''}(t) =0,~\forall e,e' \in \mathcal{E},~e \neq e'.
%\end{equation}

By studying the curves in Fig.~\ref{fig:subfigureExample}, the following arguments can be made:
\begin{itemize}
\item 
The minimum number of required measurements for $\Psi_{tot}(t)$ to achieve a worst case tail probability as a perfect i.i.d. Gaussian measurement matrix is in the same order as that of i.i.d. Gaussian (the logarithmic difference between the number of measurements for QNC and Gaussian cases is less than $1$).
Therefore, the number of required measurements in our QNC, for an {overwhelming} probability of RIP satisfaction (Eq.~\ref{Eq:GoverWprob}) is in the same order as that of an i.i.d. Gaussian matrix.
Furthermore, this behavior is improved when the number of edges in the network increases, or the corresponding RIP constant is increased.

\item 
By applying Theorem~\ref{theorem:RIP1}, on the resulting $\textbf{p}_{tail}$, the lower bound on the RIP satisfaction for each sparsity, $k$, can be obtained.
Therefore, \emph{{as an implication of RIP}} (Theorem~4.1 in \cite{naba}), we can make the following probabilistic statement about $\ell_1$-min recovery error, in QNC scenario:

\textit{Consider the QNC scenario, described in Theorem~4.1 of \cite{naba}, in which we transmit $k$-sparse messages.
In such a scenario, if the resulting $\Psi_{tot}(t)$ corresponds to a point with $\textbf{p}_{tail}(\Psi_{tot}(t),\frac{\delta_{2k}}{\sqrt{2}})$ on one of the evaluated curves of Fig.~\ref{fig:subfigureExample}, then the $\ell_2$-norm of recovery error, using the $\ell_1$-min decoder of Eq.~14 in \cite{naba}, is upper bounded according to (15) in \cite{naba}, with a probability exceeding $\textbf{p}_{RIP}\Big(\Psi_{tot}(t),2k,\delta_{2k} \Big)$.}

\item
By calculating the lower bound for RIP satisfaction (using Eq.~\ref{Eq:lowerRIPbound1}), corresponding to one of the points on the curves of Fig.~\ref{fig:subfigureExample}, it would be clear that the possible sparsity, $k$, for which the resulting $\textbf{p}_{RIP}(\Psi_{tot}(t),2k,\delta_{2k})$ approaches $1$, is very small.
In other words, QNC requires a lot of measurements to guarantee the upper bound of (15) in \cite{naba}, with an overwhelming probability.
However, this is also the case for i.i.d. Gaussian matrices, as it has been previously pointed out by the authors of \cite{phaseTransitions,sharpRIP,bah2010improved}, that the RIP analysis for i.i.d Gaussian matrices proposes an exaggerated minimum number of measurements, required for robust $\ell_1$-min recovery.
In conclusion, the minimum number of measurements, required for guaranteeing robust $\ell_1$-min decoding, using our proposed $\Psi_{tot}(t)$, is in the same order as that of i.i.d. Gaussian matrix. 
The aforementioned fact (on exaggerated required number of measurements) can be considered as a weakness of RIP analysis, used in the compressed sensing literature.

\end{itemize}

\section{Conclusions}
\label{sec:Conclusions}
Joint distributed source coding and network coding of sparse messages with compressed sensing perspective was discussed in this paper.
We investigated the satisfaction of RIP, in a modified random linear network coding scenario, called quantized network coding.
This was explicitly done by using mathematical derivation for the tail probability of the resulting measurement matrix in our QNC scenario, and that of i.i.d. Gaussian matrix.
It was numerically shown that our linear measurements have the same RIP behavior (in terms of order of minimum number of required measurements) as i.i.d. Gaussian measurements.
Our RIP analysis provided us with the preliminaries for guaranteeing robust $\ell_1$-min decoding, in QNC scenario.

\section*{Acknowledgement}
This work was supported by Hydro-Québec, the Natural Sciences and Engineering Research Council of Canada and McGill University in the framework of the NSERC/Hydro-Québec/McGill Industrial Research Chair in Interactive Information Infrastructure for the Power Grid.

\bibliographystyle{ieeetr}
\bibliography{Ref}
\end{document}